\documentclass{article}

\usepackage{arxiv}

\usepackage[utf8]{inputenc} 
\usepackage[T1]{fontenc}    
\usepackage{hyperref}       
\usepackage{url}            
\usepackage{booktabs}       
\usepackage{amsfonts}       
\usepackage{nicefrac}       
\usepackage{microtype}      
\usepackage{lipsum}
\usepackage{graphicx}
\graphicspath{ {./images/} }
\usepackage{algorithm}
\usepackage{algorithmic}
\usepackage{amsmath}
\usepackage{xcolor}
\usepackage{amsfonts}
\usepackage{bm}
\usepackage{newfloat}
\usepackage{listings}
\usepackage{amssymb}
\usepackage{natbib}
\floatstyle{ruled}
\newfloat{listing}{tb}{lst}{}
\floatname{listing}{Listing}

\usepackage{booktabs}
\usepackage{placeins}
\usepackage{amsthm}

\newtheorem{theorem}{Theorem}

\theoremstyle{definition}

\theoremstyle{remark}

\title{History Matters: Meta-policy Delegation with Heterogeneous Multi-agent Reinforcement Learning}

\author{
 Ziqing Lu\thanks{These authors contributed equally to this work.}\\
  Department of Applied Mathematics and Computational Science\\
  University of Iowa\\
  \texttt{luz@wfu.edu} \\
   \And
 Avinash Reddy Mudireddy$^*$ \\
  Department of Electrical and Computer Engineering\\
  University of Iowa\\
  \texttt{avinashreddy-mudireddy@uiowa.edu}
  \And
 Sarra Alqahtani \\
  Department of Computer Science\\
  Wake Forest University\\
  \texttt{sarra-alqahtani@wfu.edu}
   \AND
    Weiyu Xu \\
  Department of Electrical and Computer Engineering \\
  University of Iowa\\
   \texttt{weiyu-xu@uiowa.edu} \\
}

\begin{document}
\maketitle
\begin{abstract}
AI agents are expected to play an increasingly important role in future decision-making systems. In this paper, we consider collaborative systems composed of heterogeneous multi-agent systems (MAS), where their members have different capabilities and operating costs. We study how agents can delegate tasks to one another so that certain research tasks can be completed effectively under resource-constrained scenarios. We first develop a multi-agent reinforcement learning-based (MARL) delegation training that enables agents to make sequential delegation decisions while minimizing the total execution cost. We then extend this approach to MARL with prescribed delegation topologies. Furthermore, we introduce two new frameworks for collaboration and delegation in multi-agent systems. The first framework proposes that an agent's policy depends not only on the current state of the underlying Markov decision process but also on the interaction history, including previous joint actions. This history-dependent formulation can improve coordination even in fully observable environments, where conventional MARL methods typically restrict policies to depend only on the current state. The second framework proposes a novel, potentially multi-dimensional monetary mechanism to facilitate the collaboration and delegation for MAS. 
\end{abstract}


\section{Introduction}
AI agents have emerged as a promising approach for solving complex, multi-step problems. Large language models (LLMs) often serve as the reasoning engines of these agents, enabling them to interpret tasks, generate plans, use external tools, and coordinate specialized models or other agents. As agentic systems become more capable, they are also increasingly likely to contain multiple models with different capabilities, ownership relationships, and operating costs.

Despite their potential, deploying LLM-based agents can be expensive. Reliance on API-based models may incur substantial and unpredictable costs due to massive token usage. Using a premium model for every decision is often unnecessary: Smaller or less expensive models may be sufficient for routine tasks, while stronger models may be needed only for difficult or high-consequence decisions. Thus, an effective multi-agent system must determine both how to solve a task and which agent should execute each step. 

Existing work addresses this cost-performance trade-off through model routing, multi-agent orchestration, and MARL-based collaboration. FrugalGPT, RouteLLM, and AutoMix select or cascade models according to expected quality, task difficulty, and invocation cost \citep{chen2023frugalgpt,ong2025routellm,aggarwal2024automix}, while Router-R1 and MasRouter use learned controllers to make sequential routing, role-assignment, and model-selection decisions \citep{zhang2025routerr1,yue2025masrouter}. MAGRPO and CoLLM instead apply MARL to improve the task-level policies of LLM agents operating within predefined collaboration protocols \citep{liu2025magrpo,liu2026collm}. Graph-based frameworks such as GPTSwarm and G-Designer optimize computational or communication connections among agents \citep{zhuge2024gptswarm,zhang2025gdesigner}. 

Our focus is complementary but distinct: we learn a delegation layer over agents that already possess pretrained task policies. Delegation is an explicit agent action rather than a decision imposed entirely by a centralized router. 
This formulation allows us to study which agent should execute, who may delegate to whom, and how execution decisions depend on agent costs and prior interactions. Accordingly, we formulate task delegation among LLM-based agents as a MARL problem, enabling the delegation layer to adaptively and efficiently assign each task to the most appropriate agent, which then executes its own pretrained policy.

Delegation also raises questions on agents' willingness to collaborate. Whether an agent accepts a costly task may depend on whether another agent previously cooperated, fulfilled a delegation, or provided compensation, rather than because of the current physical state alone. Although histories are commonly used in partially observable environments to infer hidden state, we consider a \textbf{different role} for history: past joint actions and inter-agent interactions are strategically relevant even when the environmental state is fully observable. Such histories can support reciprocity, trust, and intertemporal incentives that may not be expressible by policies restricted to the current Markov state that were traditionally considered in fully observable MGs \cite{robust-marl,zero-sum-game,bai2020provable}.

Thus it is natural to ask ``Which agents should be trusted to perform the tasks using their own pre-trained policies?''  A related question is that ``When should agents trust their peers and delegate the work of solving the tasks to their peers?'' In this paper, we investigate how agents should delegate tasks to other agents in a collaborating environment. Our contributions include the following:

1. Our work formulates delegation itself as a sequential multi-agent decision problem. Agents with heterogeneous pretrained capabilities choose whether to execute, delegate, or remain inactive; We propose a formulation in which AI agents delegate tasks to their peers subject to management and delegation constraints. The framework is designed to reduce operational costs and may also support the development of safer and more reliable multi-agent systems \cite{oversight_game}.

2. We introduce a history-dependent policy formulation tailored to collaboration and delegation. The policies may depend on the public interaction history, including previous joint actions and states, rather than only on the current environmental state. We show that conditioning on previous behavior can support reciprocal collaboration outcomes.

3. To make history-dependent policies practical in MARL, we propose a virtual money-transfer framework that supports collaboration and delegation. An agent's money balance provides a compact summary of its past interactions, including previous actions and transfers. The virtual money is separate from the environmental rewards that measure the agent's actual task performance. By tracking the balances and constraining transfers, the system can encourage collaborations that lead to higher environmental rewards across time steps.

\textbf{Related works}: 
\cite{oversight_game} formulates an oversight game in which human supervisors decide when to intervene and when to rely on a supervised agent’s pretrained policy during MARL training. Its human–agent delegation mechanism is designed for safety guarantees. In contrast, our work studies delegation among agents as a cost-allocation problem, with the primary objective of minimizing execution costs while maintaining effective task performance.  

\cite{action_trading} studies action trading, where selfish agents can trade their actions in exchange for reward, and shows empirically that such action trading can increase the overall reward. Our money transfer framework differs in two aspects: First, the virtual money is separate from environmental rewards and records past interactions. This distinction is useful when task-level rewards are not transferable, while an internal budget can still be exchanged to facilitate delegation. Second, payments depend on the joint action, rather than only on the receiver's individual action. 

The formulation of our transferable money system is also different from the series of peer-incentivation (PI) works because our formulation treats money as a component of the state variable, while PI type of works, such as \textit{DRIVE} \cite{DRIVE} and \textit{LIO} \cite{LIO}, usually treat incentives primarily as reward-shaping signals used for policy learning. 

A line of reputation-based MARL research, including \cite{anastassacos2021cooperationreputationdynamicsreinforcement, ren2025bottomupreputationpromotescooperation}, also incorporates the effect from past interactions. However, the information resulted by past interactions in those works is either formulated as an explicit social norm or a learned social evaluation, and the resulting reputation score directly reshapes its environmental reward. In contrast, our money variable is neither a judgment of an agent’s social standing nor a reward-shaping term. It is a transferable and budget-constrained balance. For more related works in 
routing across LLMs and MARL for LLM collaboration, please see Appendix \ref{app:related_works}.

\section{Problem Formulation}
Consider a collaborative Markov game (MG) with $n$ agents: 
\begin{align}\label{model:delegation-game}
    \mathcal{M}=\{\mathcal{N},\mathcal{S},\mathcal{A}, P, R, \gamma, \{\mu_i\}_{1=1}^n, T, K\},
\end{align}
 where $\mathcal{N}$ is the number of agents with $|\mathcal{N}|=n$, $\mathcal{S}$ is the shared state space, $\mathcal{A}=\mathcal{A}_1\times\dots\times\mathcal{A}_n$ is the joint agent-group action space, $P:\mathcal{S}\times\mathcal{A}\times\mathcal{S}\rightarrow [0,1]$ is the transition probability function that depends on the joint action of the group, $\gamma\in[0,1)$ is the discount factor, $T$ is the number of steps in each episode and $K$ is the number of episodes. Each agent is individually pre-trained in the same environment and equipped with its own learned policy $\mu_{i}:\mathcal{S}\times\mathcal{A}_i\rightarrow[0,1],i=1\dots n$. The stronger the agent $i$ is in intelligence, the better its pre-learned policy $\mu_i$.
 
We consider this collaborative MG built on top of these pre-trained policies. This MG can be viewed as an extra ``wrapper'' layer between the pre-training of LLM agents and the deployment of these agents. The output of this MARL interface is $\bm{\pi}:\mathcal{S}\times\mathcal{A}\rightarrow [0,1]$, which is the joint group policy used to delegate tasks to specific agents. Notice that the delegation policy $\bm{\pi}$ is independent of the pre-trained policies $\{\mu_i\}_{i=1}^n$. 

Define each agent's one-step reward as $R_i:\mathcal{S}\times\mathcal{A}_i\rightarrow \mathbb{R}$. Then the collaborative one-step reward is $R(s_t,\bm{a}_t)=\sum_{i\in\mathcal{N}}R_i(s_t,\bm{a}_t)$. The collaborating group aims to maximize the objective function
 \begin{align}\label{model:objective}     J(\bm{\pi})=\mathbb{E}_{\bm{\pi},P}\Big[\sum_{t=1}^T\gamma^{t-1}\sum_{i\in\mathcal{N}}R_i(s_t,\bm{a}_t)\Big].
 \end{align}

On top of the collaborative game, we further introduce different delegation structures that are applicable in different scenarios. First, the \textbf{management constraints}, meaning that some agents have ownership of or control over the actions of fellow agents. Second, the \textbf{delegation constraints}, which generalize management constraints by allowing any agent to delegate tasks to any other agent. Third, the most general and realistic setting is \textbf{money trading}, in which agents incentivize others to undertake delegated tasks by allocating a portion of their held money as payment. 

\subsection{Management Constraints}\label{subsec:manage}
In the simplest case, assume that the stronger agent has one-way control over its fellow agent, and thus can mandatorily delegate work through this management structure. We take a $3$ agent system as an example. Name the three agents $A, B$, and $C$. Suppose that agent $A$ has control over agent $B$, and agent $B$ has control over agent $C$. Although $A$($B$) is a more capable agent than $B$($C$), $B$($C$) is more cost-efficient in simple tasks than its corresponding superior. Therefore, when dealing with simple tasks, superiors have incentives to delegate tasks to their fellow agents. Regardless of whether each agent chooses to delegate or adopt its pre-trained action, there is only one agent actually executing at each step.

We assume that agents have heterogeneous capabilities and execution costs. At each time step, the current state contains sufficient information for agents to identify the task. Each agent knows everybody's capabilities and costs for the given task. Thus, before choosing whether to execute or delegate, agent $i$ can determine which agents are able to complete the task and can compare their expected performance and execution costs. In this formulation, we adopt the full-information assumption and treat these profiles as known to all agents.

We formulate the collaborative delegation game as follows $\{\{A,B,C\}, \mathcal{S},\mathcal{A},R,P,\gamma, \{\mu_A, \mu_B, \mu_C\}\}$: The shared state space $\mathcal{S}$ is the same state space where agents were previously trained. For every agent $i$, its action space is augmented by its own delegation action space and the no-operation action, 
$\mathcal{A}_i'= \mathcal{A}_i\cup\mathcal{D}_i\cup\{\operatorname{no-op}\}$, with $\mathcal{D}_i$ defined as
$$\mathcal{D}_A=\{d_{A\rightarrow B}\}, \mathcal{D}_B=\{d_{B\rightarrow C}\},\mathcal{D}_C=\emptyset.$$
Each delegation action $d_{i\rightarrow j}$ means the agent $i$ delegates the tasks to agent $j$.
     The expanded joint action space is the product of expanded action spaces  $\mathcal{A}=\Pi_{i=A,B,C}\ \mathcal{A}_i'\cup\mathcal{D}_i\cup\{\operatorname{no-op}\}$. The transition probability function, discount factor are defined the same way as those of (\ref{model:delegation-game}).
 

$\{\mu_A,\mu_B,\mu_C\}$ are the pre-learned policies. For example, for each LLM agent, $\mu_i(\tilde{a}_i|s)$ denotes its fixed task-execution policy. This policy is induced by the agent's underlying pre-trained language model and available tools. Given a task state or context $s$, $\mu_i$ determines a distribution over the responses or environmental actions available to agent $i$. 

 \textbf{Sequential Execution Rules}: At each time step $t$, there is only one agent would actually execute the delegated task.
    If agent $i$ selects to perform, its action is sampled from its own pre-trained policy $a_{\mu_i}\sim \mu_i$. The execution action $a_{exec}$ of the group is chosen in the following way:
    \begin{align*}
        a_{exec}=\begin{cases} a_{\mu_A}, \text{ if }\bm{a}=(a_{\mu_A},\cdot,\cdot),\\
        a_{\mu_B}, \text{ if } \bm{a}=(d_{A\rightarrow B},a_{\mu_B},\cdot),\\
        a_{\mu_C}, \text{ if } \bm{a}=(d_{A\rightarrow B},d_{B\rightarrow C},a_{\mu_C}).
        \end{cases}
    \end{align*}
    Here ``$\cdot$'' means any action of the corresponding agent. This execution rule follows exactly the hierarchical relations between agents $A$, $B$ and $C$.

    The reward function of each agent $R_i$ consists of two parts: the environmental reward $\bar{R}$ and the individual execution or delegation cost $\mathcal{C}_i$. 
    The environmental reward $\bar{R}:\mathcal{S}\times\mathcal{A}_{exec}\rightarrow\mathbb{R}$ is generated in correspondence to the executed action $a_{exec}$, and is the same for any agent $i$ in the group. 
    
    The other part of the reward function is the individual execution or delegation cost $\mathcal{C}_i:\mathcal{A}_i\cup\mathcal{D}_i\cup\{\operatorname{no-op}\}\rightarrow\mathbb{R}.$ If agent $i$ chooses to execute, then the execution cost is generated according to the specific execution cost $\mathcal{C}_i(a_{\mu_i})$ (For example, the execution cost of each agent can be represented by the API usage cost of the corresponding LLM.);
    If agent $i$ chooses to delegate the task, there is a relatively smaller delegation cost $\mathcal{C}_i(d_{i\rightarrow j})$; If $\operatorname{no-op}$ is chosen, $\mathcal{C}_i(\operatorname{no-op})=0,\forall i\in\mathcal{N}$.
 
 Therefore, one-step reward function of each agent $i$ is formulated as 
$$R_i(s,\bm{a})=\alpha\bar{R}(s,a_{exec})-\beta\mathcal{C}_i(a_{i}),$$
where $\alpha$ and $\beta$ are scaling factors, since the environmental reward and the money cost differ in units. The objective of the agent group is still to maximize (\ref{model:objective}).

\subsection{Delegation Constraints}\label{subsec:delegate}
In a more general case, when delegation is assigned in interactive ways between every two pairs of agents, delegation constraints can be represented through a directed graph. For example, an arrow from Agent $A$ to Agent $B$ represents $A$ choosing to delegate to $B$, and vice versa. Compared to subsection \ref{subsec:manage}, we omit the strict hierarchical relation in delegation rules.


With the same three-agent example as in \ref{subsec:manage}, we let all other settings be the same but modify the action spaces and the execution rules 

The individual delegation action space $\mathcal{D}_i$ of each agent $i$ is now expanded, and any agent can delegate the task to any other agent in the group.
    $$\mathcal{D}_{A}=\{d_{A\rightarrow B}, d_{A\rightarrow C}\}, \mathcal{D}_{B}=\{d_{B\rightarrow A}, d_{B\rightarrow C}\}, $$
    $$\mathcal{D}_{C}=\{d_{C\rightarrow A}, d_{C\rightarrow B}\}$$
   \textbf{Networked Execution Rules}: Multiple agents can execute at the same time. The execution action is now a $3$-element vector $\bm{a}$, with each element indicating whether and with what action an agent is executing. Here $(\bm{a}_{\mathrm{exec}})_i$ denotes the $i$-th element of the execution action.
    \begin{align*}
     (\bm{a}_{\mathrm{exec}})_i = a_{\mu_i}, \text{if} \begin{cases}
           \exists\, j\neq i:\ a_j=d_{j\rightarrow i}
       \ \\
        a_i\sim \mu_i.
       \end{cases}
    \end{align*}
    This execution rule means that an agent performs if and only if some other agent delegates the task to it and it does not delegate to any other agent. For initialization, at the beginning of every episode, an oracle delegates the task to a selected subset of agents. The individual reward remains the same as that of \ref{subsec:manage}, and the group objective is the same as (\ref{model:objective}).

In realistic human societies, whether a delegation relationship can be formed often depends on an individual's willingness to accept a costly task. To capture this feature, we allow each agent's willingness to accept such a task to depend on its previous interactions with other agents. In particular, an agent may be more willing to trust a peer and accept a delegated task if they have chosen collaborative joint actions in the past. We therefore model each agent's policy as depending on the history of past states and joint actions, rather than only on the current state. 

\subsection{Policies Dependent on Past Interactions}\label{subsec:history-dependent policies}
In this subsection, we consider a class of history-dependent policies that differs from the conventional Markov policies commonly used in the MARL literature \cite{robust-marl,zero-sum-game,bai2020provable}. Under a Markov policy, each agent selects its action based only on their current state. By contrast, under a history-dependent policy, an agent may condition its action on the entire past interactions of the game.

Consider a finite dynamic game with $n$ participants
$\mathcal{G}=\{\mathcal{S},\mathcal{A}_1\times \cdots \times \mathcal{A}_n,P,T,r\}.$ Let $x_t\in \mathcal{S}$ denote the state at time $t$, and let $\bm{a}_t=(a_t^1,\dots,a_t^n)$ denote the joint action at time $t$, $1\leq t\leq T$. We define the public history available at time $t$ by
$h_t=(x_1,\bm{a}_1,x_2,\bm{a}_2,\dots,x_{t-1},\bm{a}_{t-1},x_t).$ The public history $h_t$ contains all past states, all past joint actions, and the current state. For each agent $i$, a deterministic history-dependent policy profile is a sequence of policies $\{\pi_t^i\}_{t=1}^{T}$, where each history dependent policy maps from the public history up to time $t$ to the action space of agent $i$, $\pi_t^i:\mathcal{H}_t\rightarrow \mathcal{A}_i,\  t=1,\dots,T,$ where $\mathcal{H}_t$ denotes the set of histories up to time $t$. A Markov policy is a special case of history-dependent policies, where the policy depends only on the current state and is irrelevant to past trajectories: $\pi_t^i(h_t)=\mu_t^i(x_t).$

The key advantage of a history-dependent policy is that it allows agents to reward, punish, or condition their actions based on the previous behavior of other agents. The following example illustrates why this additional dependence on history is more realistic and can even incentivize equilibrium policies with higher returns. Notice that the regular Markov policy space is a subspace of the history-dependent policy space, so the return of a history-dependent policy is at least as good as that of a Markov policy.  

\textbf{Motivating example}: 
Consider a two-agent dynamic game with two agents $A$ and $B$, with two time steps $T=2$, and with a state space $\mathcal{S}=\{s_1,s_2\}$. The state at time $t$ is denoted by $x_t$. The action spaces are $\mathcal{A}_1=\{a_{A,1},a_{A,2}\},\ 
\mathcal{A}_2=\{a_{B,1},a_{B,2}\}$. The joint action is written as $\bm{a}_t=(a_t^A,a_t^B)$.

At the first time step, agent $A$'s action determines the reward:
\begin{align*}
    r_1(x_1, \bm{a}_1)=
\begin{cases}
(0,0), & \text{if } a_1^A=a_{A,1},\\
(-0.1,1), & \text{if } a_1^A=a_{A,2},
\end{cases}
\end{align*}
where $a_{A,2}$ is costly for agent $A$, but beneficial for agent $B$.

At the second time step, agent $B$'s action determines the reward:
\begin{align*}
    r_2(x_2,\bm{a}
    _2)=
\begin{cases}
(0,0), & \text{if } x_2=s_1 \text{ and } a_2^B=a_{B,1},\\
(9,0), & \text{if } x_2=s_1 \text{ and } a_2^B=a_{B,2},\\
(0,10), & \text{if } x_2=s_2 \text{ and } a_2^B=a_{B,1},\\
(9,10), & \text{if } x_2=s_2 \text{ and } a_2^B=a_{B,2}
\end{cases}
\end{align*}
where $a_{B,2}$ benefits agent $A$, while agent $B$'s own second-step reward is independent of whether it chooses $a_{B,1}$ or $a_{B,2}$. The state transition is given by $\mathbb{P}(x_2=s_2\mid x_1=s_1,\bm{a}_1)=0.5,
\mathbb{P}(x_2=s_1\mid x_1=s_1,\bm{a}_1)=0.5, \forall \bm{a}_1\in\mathcal{A}$.

\textit{History-dependent Nash equilibrium.} Under the history-dependent policy formulation, the following policy profile is one deterministic Nash equilibrium:
\begin{align*}
    &\pi_1^A(h_1)=a_{A,2},
\qquad
\pi_2^A(h_2)\text{ any action in }\mathcal{A}_1,\\
&\pi_1^B(h_1)\text{ any action in } \mathcal{A}_2, \pi_2^B(h_2)=
\begin{cases}
a_{B,1}, & \text{if } a_1^A=a_{A,1},\\
a_{B,2}, & \text{if } a_1^A=a_{A,2}.
\end{cases}
\end{align*}
This policy is history-dependent because agent $B$'s second-period action depends on the past joint action $\bm{a}_1$, not only on the current state $x_2$.

Under this policy profile, agent $A$ chooses the costly action $a_{A,2}$ at the first time step. In response, agent $B$ chooses $a_{B,2}$ at the second time step, which rewards agent A. The expected returns are
$V_0^A=-0.1+0.5(9+9)=8.9,$ and $V_0^B=1+0.5(0+10)=6.$

\textit{Markov Nash equilibrium.}
Now suppose that agents are restricted to Markov policies. Then agent B's second-period policy can depend only on the current state $x_2$, but not on agent A's first-period action. Therefore, the history-dependent reward mechanism described above is no longer available.

An example Markov equilibrium is given by 
\begin{align*}
    &\pi_1^A(x_1)=a_{A,1},
\qquad
\pi_2^A(x_2)\text{ can be any action in }\mathcal{A}_1,\\
&\pi_1^B(x_1)\text{ can be any action } \mathcal{A}_2,
\qquad
\pi_2^B(x_2)=a_{B,1}.
\end{align*}

Under this Markov policy profile, agent $A$ does not choose the costly first-period action, since agent $B$'s future action cannot be conditioned on whether A previously helped $B$. The expected returns are $V_0^A=0+0.5(0+0)=0,$ and $V_0^B=0+0.5(0+10)=5.$ In fact, we cannot find any equilibrium for history-independent policies where agent $B$ achieves a reward bigger than $6$.

We numerically verified that among all the Markov deterministic equilibria policies, there are $3$ possible outcomes in terms of value: $\{(0,0,5.0),(4.5, 5.0), (9.0, 5.0)\}$, where within each tuple, the first element is the return of agent $A$ and the second element is the return of agent $B$. Meanwhile, among all the history-dependent deterministic equilibria, there are $5$ possible outcomes:
$\{(0.0,5.0), (4.4,6.0),(4.5,5.0),(8.9,6.0),(9.0,5.0)\}$.

This example illustrates that history-dependent policies can support reciprocity mechanisms that are not expressible by Markov policies. Since future actions can depend on past joint actions, one agent may be incentivized to take a costly action early in the game if another agent can condition its later behavior on that action. In this sense, history-dependent policies may support equilibria with higher returns than those obtained under the Markov policy restriction.

\textbf{Algorithm}: We develop a \textit{Two-step Nash 
Value Iteration Algorithm} to find the overall deterministic Nash equilibrium policy that favors agent $B$. In the case when there are multiple pure Nash equilibria, depending on the agents' need, Algorithm \ref{alg:time2-nash-vi} outputs the policy profile that favors the return of the chosen agent $i$. Please refer to Appendix \ref{app:algorithm} for details of this algorithm. We prove that the output of this algorithm is a history-dependent pure-strategy Nash equilibrium.

\begin{theorem}
\label{thm:optimal-selector}
Consider a finite two-step dynamic game. If the two-step continuation game admits at least one pure-strategy Nash equilibrium, and that at least one induced one-step game admits a pure-strategy Nash equilibrium, then Algorithm \ref{alg:time2-nash-vi} returns a pure-strategy subgame-perfect Nash equilibrium policy profile.
\end{theorem}

\begin{proof}
We prove by showing that none of the agent $i$ prefers to deviate from the output policy $\pi^{*}$ at any time step $t=1$ and $t=2$. Please see Appendix \ref{app:proof} for details of this proof.
\end{proof}

\textbf{Extension of Algorithm \ref{alg:time2-nash-vi} to Finite $n$-step Games}: The two-step algorithm can be extended to an $n$ step dynamic game by induction. First, solve the final two-step continuation game, consisting of times $n-1$ and $n$ using Algorithm \ref{alg:time2-nash-vi}. This generates the set of deterministic history-dependent Nash equilibrium continuation policies for the last two steps.

Next, apply each continuation policy for the previous step $n-2$. For every possible history $h_{n-2}$, compute the immediate reward plus the expected equilibrium value generated by the selected continuation policy. Then find the Nash equilibria of this induced one-step game and combine each equilibrium action with the continuation policy. This gives the equilibrium policies for the last three steps. Repeat this procedure backward until the algorithm yields the deterministic history-dependent Nash equilibrium policies for the entire $n$-step dynamic game.

\subsection{Delegation Through Transferable Money}\label{subsec:reward_trade}
We now consider a more realistic setting in which agents can form delegation contracts through transferable money. The key feature of this framework is that money is separate from the environmental rewards that agents aim to maximize. The reward function represents the real-world utility obtained by each agent from the environment, while money only serves as an internal medium as a record of past joint actions and states, and money enables delegation and coordination in a distributed multi-agent reinforcement learning. 

Consider a Markov game (\ref{model:delegation-game}). Let $s_t$ denote the environmental state at time $t$, and let $\mathcal{A}_i$ be the environmental action space of agent $i$. To improve this regular Markov game with a money transfer mechanism, each agent is equipped with a money balance state $w^i_t\geq 0,\forall t,i$, and define the augmented individual state as $\tilde{s}^i_t=(s_t,w^i_t)$. Now the individual state records both the environmental state and the leftover money balance of agent $i$. More generally, the money can be multi-dimensional instead of a scalar.

At each time step $t$, each agent $i$ chooses a physical action $a_t^i$ and a money transfer action $\bm{p}_t^i$.
This physical action is mapped from the environmental policy ${a}^i_t\sim\pi_i^{\operatorname{env}}(\tilde{s}^i_t)$.

\textbf{Money Transfer Actions} The money transfer action of each agent $i$ is $\bm{p}_t^i=[p^{ij}_t]_{j\neq i}\in \mathbb{R}^{n-1}_{+}$. Its domain is determined by a prescribed contract: $\mathrm{Cr}^i:\mathcal{A}\rightarrow \mathbb{R}^{n-1}$ that maps from each \textit{joint action} to a list of fixed prices agent $i$ is willing to pay for other agents. At time step $t$, if agent $j$ performs the corresponding physical action $a_j$ at step $t$, agent $i$ pays money $p^{ji}_t=\mathrm{Cr}^i_j(\bm{a})$ to agent $j$ (we can also make the amount depend on historical actions).
Thus, the discrete money transfer actions $\bm{p}^{i}_t$ are chosen according to the current augmented state and the current joint action (it can also depend on historical actions or trajectories):
$\bm{p}^{i}_t\sim \pi^{\operatorname{tr}}_i(\cdot|\tilde{s}_t, \bm{a}_t).$   
Agent $i$'s transfer policy is therefore denoted by $$\pi_i^{\operatorname{tr}}:\tilde{s}_t\times\mathcal{A}\rightarrow\mathbb{R}^{n-1}_{+}.$$
This policy determines the vector of payments $[p^{ij}_t]_{j\neq i}$ made by agent $i$ to all other agents.  The full policy of agent $i$ is therefore $\pi_i=(\pi_i^{\operatorname{env}},\pi_i^{\operatorname{tr}})$, and the joint policy is  $\bm{\pi}=(\bm{\pi}^{\operatorname{env}},\bm{\pi}^{\operatorname{tr}}).$

\textbf{Insufficient Money Balance} When agent $i$ wants to pay more than its current balance, i.e., $\| \bm{p}_{t}^{i}\|_1 > w_t^i$, then the desired joint action is not feasible, and agent $i$ needs to turn to an affordable joint action. When $w_t^i$ drops to $0$, then no transfer is allowed. 

The agent's objective is still purely environmental: $J(\bm{\pi})=\mathbb{E}^{\bm{\pi}}[\sum_{t=1}^{T}\sum_{i=1}^N \gamma^{t-1}R_i(s_t,\bm{a}_t)],$
and money does not directly enter the reward function. 

\section{Numerical Analysis}
\label{sec:numerical}
This section presents the numerical results tied to the mechanisms introduced in the problem formulations.
We begin with management and delegation structures on MARL. The delegation policy is trained for solving GSM8K~\citep{cobbe2021gsm8k}, a
public grade-school math word-problem benchmark with reference answers, using
one shared solver stack and changing only the interaction rule between the two
subsections.
Protocol knobs that matter for reproducibility---prices, caches, seeds,
hyper-parameters, and compute---are collected in Appendix~\ref{app:gsm8k-suppl}.
The remaining subsections treat history-dependent policies and money transfer
separately.

\paragraph{Shared GSM8K stack.}
We use three frozen solvers of different strength and cost on this math
benchmark: Nemotron~3 Ultra as solver~$A$,
Llama~3.2~3B as solver~$B$, and Llama~3.2~1B as solver~$C$.
Only a centralized DQN coordinator $\pi$ is trained; $A$, $B$, and $C$ stay
frozen.
Synthetic token prices make solver~$A$ expensive and solver~$C$ cheap
(Appendix~\ref{app:gsm8k-suppl}).
Before training the MARL interface, we cache each solver's GSM8K answers and offline correctness labels
once; during training $\pi$ reuses those fixed offline answers and labels
instead of calling the solvers or a live labeler at every step, which keeps the
runs reproducible (Appendix~\ref{app:gsm8k-suppl}).
Each episode packs $32$ questions, allows up to three attempts per question
(fail penalty $-1$), and runs for $1000$ episodes with
$\gamma=0.99$, cost weight $\lambda=100$, and handoff fee $\alpha=0.05$.
Unless noted otherwise, we report mean$\pm$std of return, accuracy, and cost
across seeds $0$--$19$.

The \emph{always-Ultra} baseline is a no-handoff policy that sends every
question to solver~$A$ alone, under the same caches and cost rules as the
learned policy in that subsection.

\subsection{Management Constraints on GSM8K}
\label{subsec:exp21}
\label{subsec:gsm8k-mgmt-protocol}
Here we put the hierarchy of Section~\ref{subsec:manage} on GSM8K in a
restricted form.
The coordinator may use only three routes along the chain
$A\rightarrow B\rightarrow C$:
(i)~solver~$A$ answers alone;
(ii)~$A$ hands the question to solver~$B$, and $B$ answers;
or (iii)~the question is passed $A\rightarrow B\rightarrow C$, and solver~$C$
answers.
Exactly one solver runs on each attempt.
Skipping straight from $A$ to $C$ (allowed in the theory section) is turned off
here, so every handoff must follow the chain in order.

The step reward credits a correct answer under the offline labels, charges for
the solve and for each handoff, and applies a failure penalty after three
misses:
$$
r \;=\; R_{\mathrm{corr}}
\;-\; \lambda\, C_{\mathrm{solve}}(a_{\mathrm{exec}})
\;-\; \lambda\, C_{\mathrm{edge}}(n_e,a_{\mathrm{exec}})
\;+\; R_{\mathrm{fail}},
$$
where $C_{\mathrm{edge}}=n_e\,\alpha\,C_{\mathrm{solve}}$ with $\alpha=0.05$,
and $R_{\mathrm{fail}}$ is $0$ on ordinary steps and $-1$ after three failed
attempts (so the leading ``$+$'' adds a negative penalty; see
Appendix~\ref{app:gsm8k-suppl}).

Table~\ref{tab:gsm8k-mgmt-train} and Figure~\ref{fig:gsm8k-mgmt-agg-with}
show the results in the late MARL training.
The \textit{always-Ultra} baseline has return $23.90{\pm}0.30$ at the cost of
$\$0.063{\pm}0.002$ per episode.
The MARL learned coordinator reaches $27.56{\pm}0.32$ return at the cost of 
$\$0.033{\pm}0.002$ per episode (about $\$31.1$ vs.\ $\$60.8$ cumulative over
$1000$ episodes), with accuracy $0.982{\pm}0.002$ versus $0.971{\pm}0.002$.
That is roughly $+3.7$ return at about half the synthetic cost.
A paired Wilcoxon signed-rank test over the $20$ seed-wise late-train means
finds higher return and lower cost than always-Ultra in all $20$ seeds
(two-sided $p=2^{-19}$). Figure~\ref{fig:gsm8k-mgmt-macro} tracks which solver is called.
Late policies lean on solver~$B$, still ask solver~$A$ on a nontrivial share,
and almost never rely on solver~$C$.

\begin{table}[t]
\centering
\scriptsize
\setlength{\tabcolsep}{3pt}
\begin{tabular}{lccc}
\toprule
Late-train metric & DQN+map & DQN no map & always-Ultra \\
\midrule
Return & $27.56{\pm}0.32$ & $27.65{\pm}0.28$ & $23.90{\pm}0.30$ \\
Accuracy & $0.982{\pm}0.002$ & $0.982{\pm}0.002$ & $0.971{\pm}0.002$ \\
Cost (\$/ep) & $0.033{\pm}0.002$ & $0.032{\pm}0.002$ & $0.063{\pm}0.002$ \\
Cum.\ cost (\$/$1000$) & $31.1{\pm}0.7$ & $31.1{\pm}0.7$ & $60.8{\pm}0.2$ \\
\bottomrule
\end{tabular}
\caption{Management Constraints: Late-training metrics (last $100$ episodes; mean$\pm$std over $20$ seeds).}
\label{tab:gsm8k-mgmt-train}
\end{table}

\begin{figure}[t]
\centering
\IfFileExists{figures/exp21/aggregate_mean_std_with_map.png}{\includegraphics[width=0.95\linewidth]{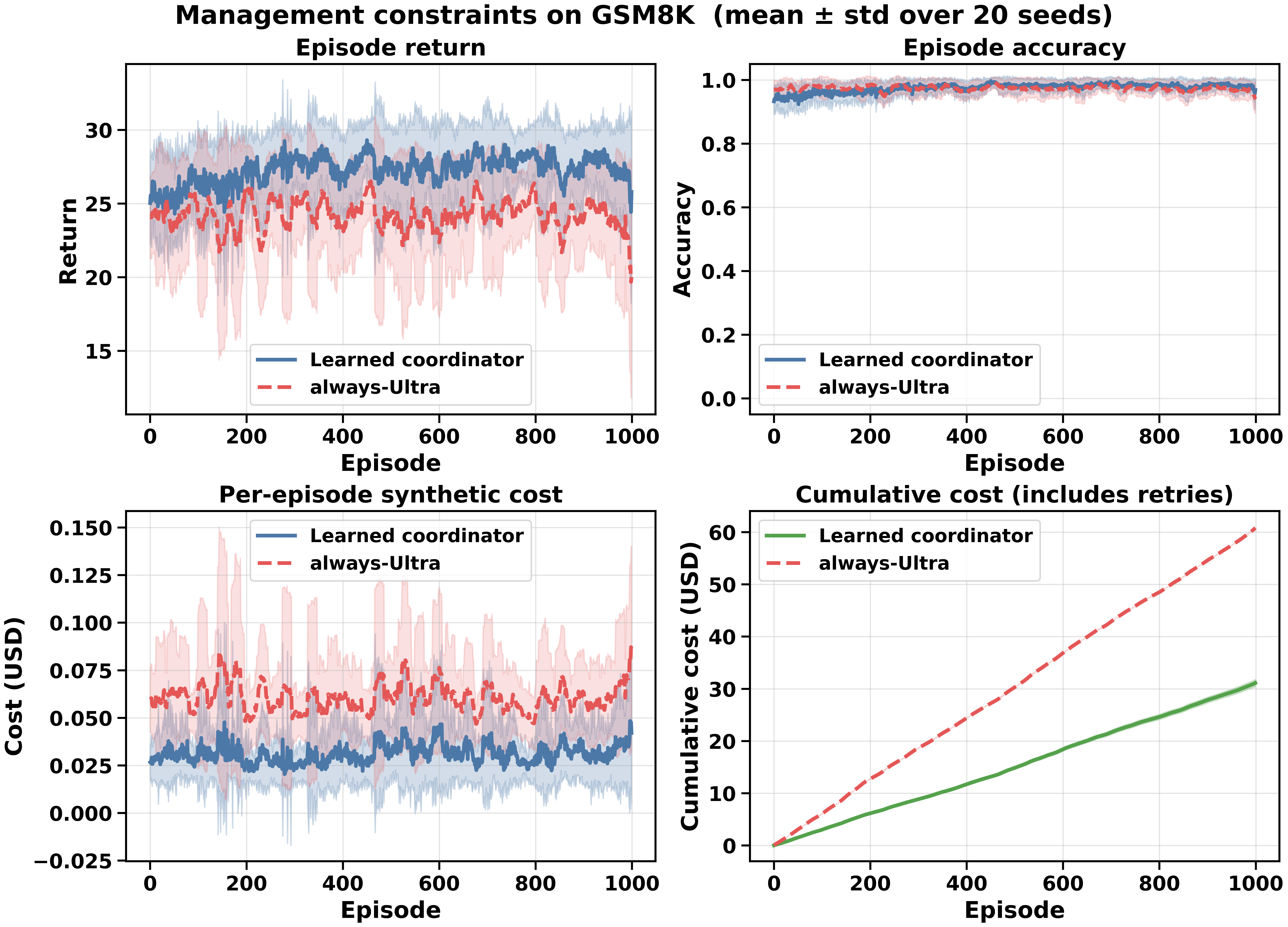}}{\vspace{1.8in}}
\caption{Management constraints on GSM8K: mean$\pm$std over $20$ seeds for return, accuracy, per-episode cost, and cumulative cost (MARL learned coordinator vs always-Ultra).}
\label{fig:gsm8k-mgmt-agg-with}
\end{figure}

\begin{figure}[t]
\centering
\IfFileExists{figures/exp21/aggregate_macro_mix_with_map.png}{\includegraphics[width=0.95\linewidth]{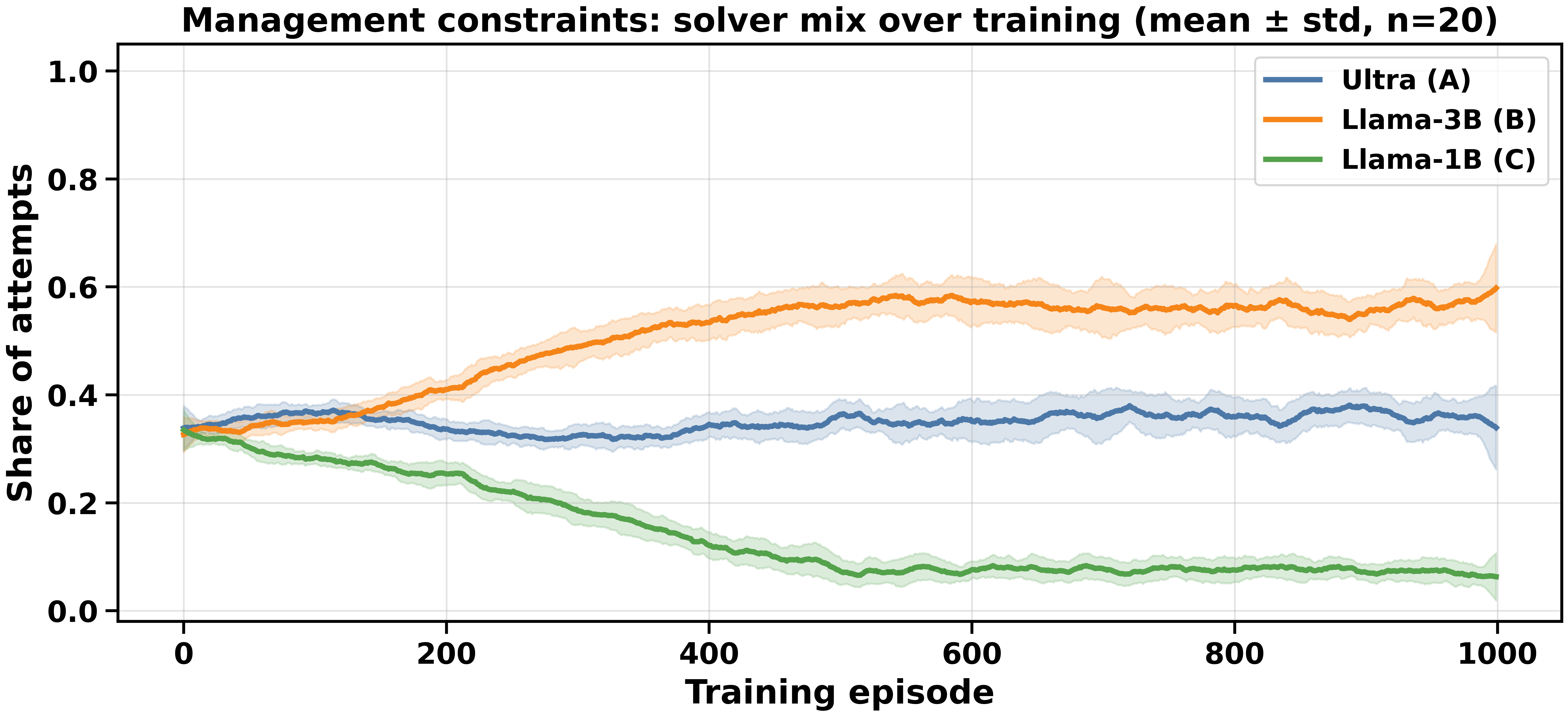}}{\vspace{1.2in}}
\caption{Management constraints: share of attempts using Ultra ($A$), Llama~3B ($B$), and Llama~1B ($C$) over training (mean$\pm$std over $20$ seeds).}
\label{fig:gsm8k-mgmt-macro}
\end{figure}

\subsection{Delegation Constraints on GSM8K}
\label{subsec:exp22}
We keep the same GSM8K stack and swap to the freer delegation rule indicated in
section~\ref{subsec:delegate}. The coordinator first chooses a non-empty kickoff set among the seven options
$\{\{A\}$, $\{B\}$, $\{C\}$, $\{A,B\}$, $\{A,C\}$, $\{B,C\}$, and $\{A,B,C\}\}$.
Then, each of $A$, $B$, and $C$ chooses one of the three actions: answer,
pass to the first other agent, or pass to the second other agent.
Packing kickoff with the three local choices gives
$7\times 3^3=189$ joint discrete actions.
An agent answers only if it was handed the task (by kickoff or by a peer pass)
and it chooses to answer; Multiple agents may answer in the same attempt.
We count the attempt as successful if any agent is correct. We charge the sum of those agents' solve costs, and add $\alpha=0.05$ edge fee (each kickoff, handoff, and peer pass).

Besides \textit{always-Ultra}, we roll out the frozen management coordinator from
section~\ref{subsec:exp21} on the same packets for a matched comparison.
\textit{Always-Ultra} is evaluated in the same graph environment, so it also pays the
kickoff edge fee when solver~$A$ is started.

Table~\ref{tab:gsm8k-del-train} and Figure~\ref{fig:gsm8k-del-agg-with}
show results in the late training.
Against \textit{Always-Ultra}, which returns $23.31{\pm}0.30$ at the cost of 
$\$0.068{\pm}0.002$ per episode, the graph coordinator reaches
$25.20{\pm}0.42$ return at the cost of $\$0.043{\pm}0.002$/episode, with accuracy
$0.961{\pm}0.005$ versus $0.971{\pm}0.002$.
There is about $1.9$ increased return at roughly $37\%$ lower per-episode cost.
The cumulative synthetic cost over $1000$ episodes falls from about $\$66.7$ to
$\$45.5$.
The same Wilcoxon test again finds a higher return and lower cost than
\textit{Always-Ultra} in all $20$ seeds with two-sided $p=2^{-19}$.
On matched packets, the management coordinator is still ahead. It returns $28.70{\pm}0.18$, with accuracy $0.986{\pm}0.001$, at the cost
$\$0.024{\pm}0.001$ per episode.
Therefore, the graph delegation beats never handing off, while the hierarchy policy remains
stronger on return at a lower cost in these runs.

Figure~\ref{fig:gsm8k-del-exec} shows the executor mix.
Compared with management routing, the graph policies call solver~$C$ more often,
and sometimes several agents answer the same question.
That freer pattern helps explain why return and accuracy sit a bit below the
hierarchy policy, even though the cost stays below \textit{Always-Ultra}.

\begin{table}[t]
\centering
\scriptsize
\setlength{\tabcolsep}{2.5pt}
\begin{tabular}{lcccc}
\toprule
Late-train metric & DQN+map & DQN no map & always-Ultra & Mgmt.\ $\pi$ \\
\midrule
Return & $25.20{\pm}0.42$ & $25.21{\pm}0.33$ & $23.31{\pm}0.30$ & $28.70{\pm}0.18$ \\
Accuracy & $0.961{\pm}0.005$ & $0.961{\pm}0.004$ & $0.971{\pm}0.002$ & $0.986{\pm}0.001$ \\
Cost (\$/ep) & $0.043{\pm}0.002$ & $0.043{\pm}0.002$ & $0.068{\pm}0.002$ & $0.024{\pm}0.001$ \\
Cum.\ cost (\$/$1000$) & $45.5{\pm}0.6$ & $45.4{\pm}0.6$ & $66.7{\pm}0.2$ & $23.9{\pm}0.1$ \\
\bottomrule
\end{tabular}
\caption{Delegation constraints: Late-training metrics (last $100$ episodes; mean$\pm$std over $20$ seeds; handoff fee $\alpha=0.05$). Mgmt.\ $\pi$ is the frozen management coordinator on the same packets.}
\label{tab:gsm8k-del-train}
\end{table}

\begin{figure}[t]
\centering
\IfFileExists{figures/exp22/aggregate_mean_std_with_map.png}{\includegraphics[width=0.95\linewidth]{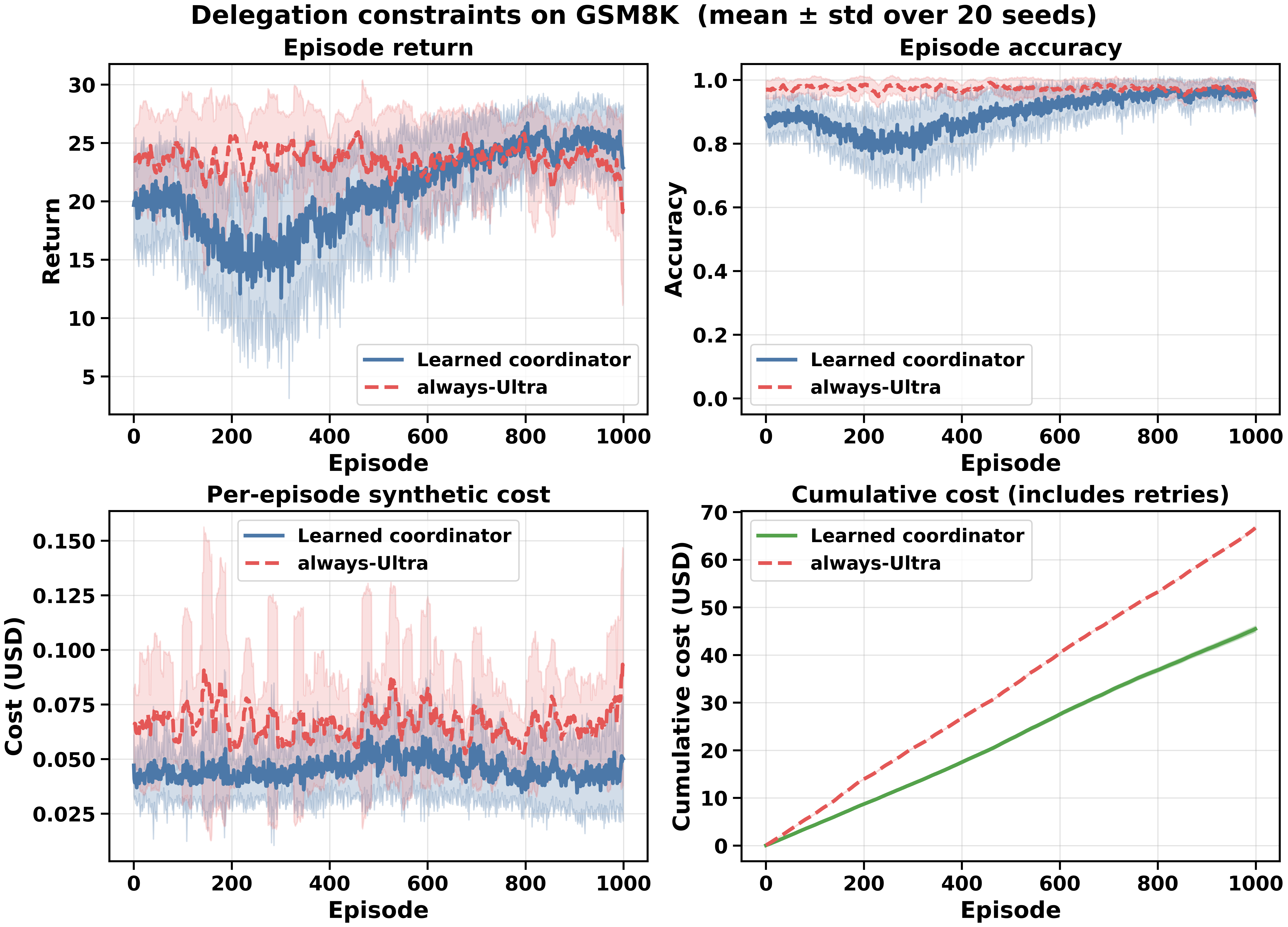}}{\vspace{1.8in}}
\caption{Delegation constraints on GSM8K: mean$\pm$std over $20$ seeds for return, accuracy, per-episode cost, and cumulative cost (learned coordinator vs always-Ultra).}
\label{fig:gsm8k-del-agg-with}
\end{figure}

\begin{figure}[t]
\centering
\IfFileExists{figures/exp22/aggregate_executor_mix_with_map.png}{\includegraphics[width=0.95\linewidth]{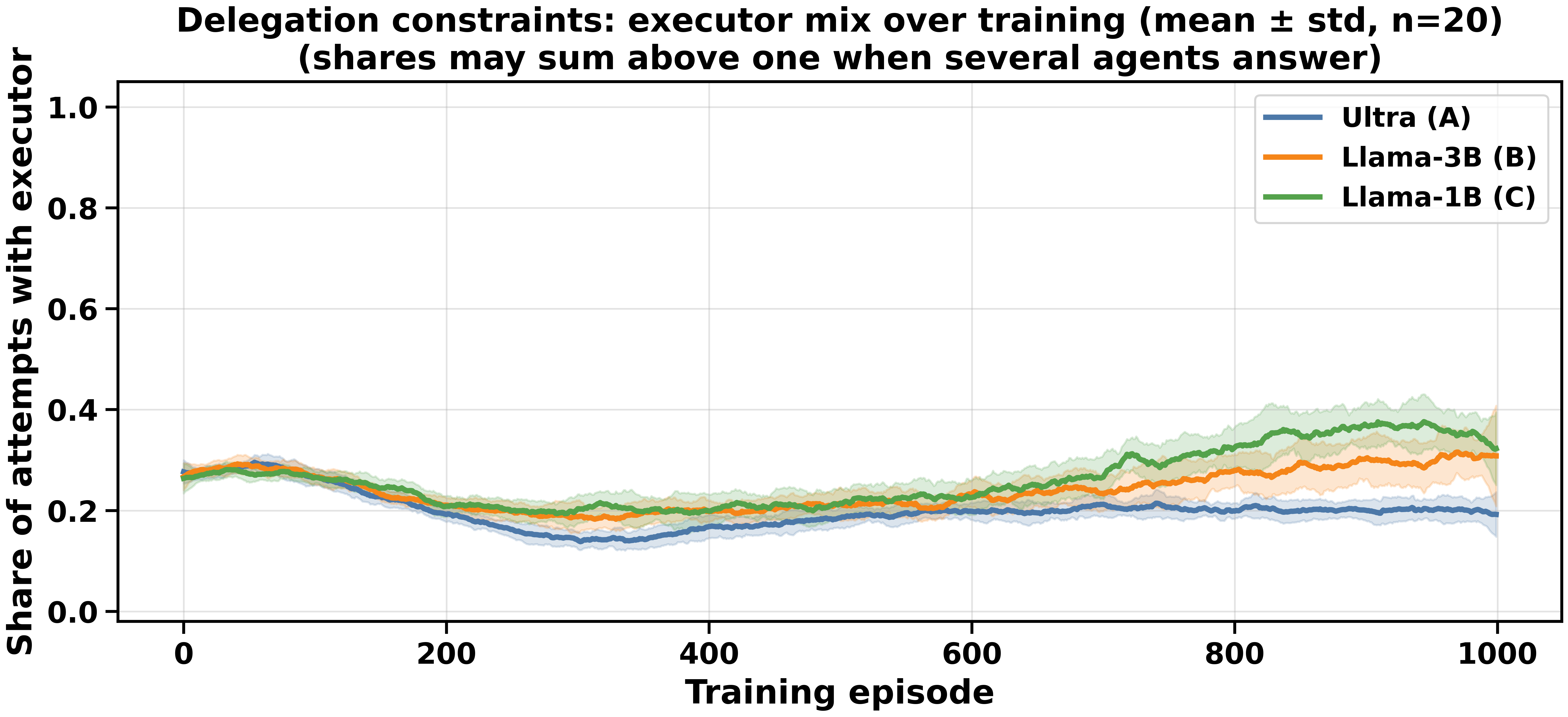}}{\vspace{1.2in}}
\caption{Delegation constraints: share of attempts with Ultra ($A$), Llama~3B ($B$), and Llama~1B ($C$) executing over training (mean$\pm$std over $20$ seeds; shares may sum above one when several agents answer).}
\label{fig:gsm8k-del-exec}
\end{figure}

\subsection{Interactions of Distributed Agents with the Money Trading Mechanism}
\label{subsec:money_trading_exp}
This experiment trains two independent Q-learning agents with and without the money transfer mechanism for comparison. The environment works as follows: There are two agents $A$ and $B$ participating in the sequential games. In the training, there are $20$ steps in each episode, and $5000$ episodes in total. The agents alternate in being the active agent: Agent $A$ is active in one phase, typically the odd-numbered steps, and agent $B$ is active in another phase, typically in even-numbered steps. However, only the action of the currently active agent affects the physical rewards. 

A state has two components $s=(\operatorname{phase},\operatorname{t})$ where $\operatorname{phase}\in\{s_1,s_2\}$ indicates which which agent is active. $\operatorname{t}$ is the current time step index. This experiment has two phases, and therefore, the state space is $\mathcal{S}=\{(s_1,\operatorname{t}),(s_2,\operatorname{t}):\operatorname{t}=0,\dots,20\}$. The action space of each agent is binary: $\mathcal{A}_A=\mathcal{A}_B=\{0:\operatorname{Null},1:\operatorname{Yes}\}$, where $\operatorname{Null}$ means ``do not help with others'', and $\operatorname{Yes}$ means ``help the other agent''. The transition of states is deterministic, such that the state alternates between the two possible phases as $s_1\rightarrow s_2\rightarrow s_1\rightarrow \dots \rightarrow s_2$ regardless of the joint actions. The time counter $\operatorname{t}$ increases by $1$ after every step within each episode.

The active agent chooses whether to incur a small personal cost to generate a large benefit for the other agent. When agent $A$ is active, the reward function is that 
\begin{align*}
(r^A_t,r^B_t)=
    \begin{cases}
    (0,0),\text{ if $A$ chooses }\operatorname{Null}\\
    (-1,100),\text{ if $A$ chooses }\operatorname{Yes}
    \end{cases}.
\end{align*}
Symmetrically, if $B$ is active, then the reward function is
\begin{align*}
(r^A_t,r^B_t)=
    \begin{cases}
    (0,0),\text{ if $B$ chooses }\operatorname{Null}\\
    (100,-1),\text{ if $B$ chooses }\operatorname{Yes}
    \end{cases}.
\end{align*}
The returns of the independent Q-learning with and without a money transfer mechanism are shown in Figure~\ref{fig:money_transfer_results}.

\begin{figure}
    \centering
    \includegraphics[width=0.85\linewidth]{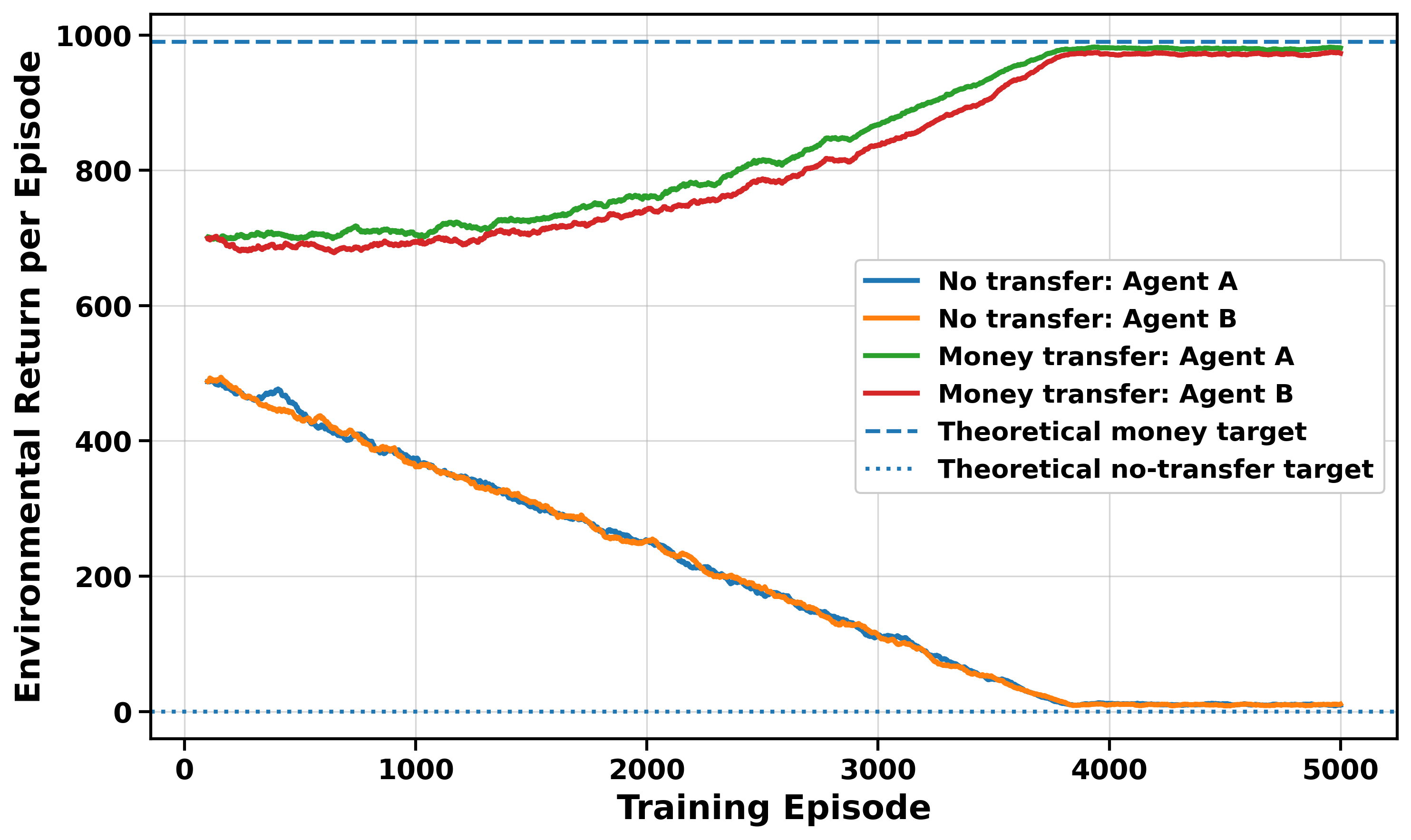}
    \caption{Training Returns of Independent Q-learning with Money Transfer Mechanism}
    \label{fig:money_transfer_results}
\end{figure}
\textit{1. Independent Q-learning (IQL)} Agents $A$ and $B$ maintain separate $Q$-functions and independently select their actions using an $\varepsilon$-greedy policy. Each agent updates its action values using only its own environmental reward, while treating the other agent's behavior as part of the environment. After the $5000$-episode training, the learned policy is that both agents refuse to collaborate. On agent $A$'s turns, agent $A$ always chooses the action $\operatorname{Null}$. On agent $B$'s turns, agent $B$ always chooses the action $\operatorname{Null}$ as well. This non-cooperative joint policy results in a $0$ episodic environmental return.

\textit{2. IQL with Money Transfer} With the new formulation, we consider an extended independent Q-learning with money transfers. The environment is augmented with transferable money balances, and therefore a state consists of $s=(\operatorname{phase},\operatorname{t},\bm{w})$, where $\bm{w}=(w_1,w_2)$ is the money balances of both agents. Each agent starts with an initial balance of $1$ unit. Before the active agent chooses its physical action, the other agent acts as the payer and decides whether to offer a fixed payment $1$ unit for the $\operatorname{Yes}$ action. An offer is made only when the payer has sufficient funds: $w_i\geq 1$. If the payment is made, the corresponding amount is transferred from the payer's balance to the active agent's balance, and the contract requires the active agent choosing the action $\operatorname{Yes}$.

After $5000$ training episodes, the learned joint policy is that on agent $A$'s turns, agent $B$ transfers one unit to $A$. By contract, agent $A$ chooses the action $\operatorname{Yes}$ to help the agent $B$. The physical and money policies are symmetric on agent $B$'s turns. For every episode with $20$ time steps, both agents helped at all $10$ of their active stages and obtained an environmental return of $990$. Thus, the money-transfer framework successfully supported the fully cooperative outcome.

\section{Conclusion}
\label{sec:conclusion}
In this paper, we study the delegation problem among a heterogeneous multi-agent group. We also proposed frameworks of collaboration/delegation for multi-agent systems with history-based policies and money-based policies. In terms of limitations, this work has not addressed the pricing of the joint actions based on history, even though we have shown that certain pricing of the joint actions can help the agents achieve collaborations. More efficient algorithms are needed for finding the Nash equilibra and coarse correlated equilbria using history-dependent policies.

\bibliographystyle{unsrt}  
\bibliography{related_work_citations, ref_part1}  

\appendix
\section{Related Work}
\label{app:related_works}
\subsection{Cost-Aware Routing Across Language Models}
A growing body of work studies how to exploit differences in the capabilities and costs of available language models. FrugalGPT considers prompt adaptation, model approximation, and cascades of language models to reduce inference cost while maintaining response quality \citep{chen2023frugalgpt}. RouteLLM learns a router from preference data to select between a stronger, more expensive model and a weaker, less expensive model \citep{ong2025routellm}. AutoMix first queries a smaller model and uses self-verification together with a partially observable decision process to determine whether the query should be escalated to a larger model \citep{aggarwal2024automix}. RouterBench and LLMRouterBench provide standardized datasets and evaluation frameworks for studying routing performance and performance--cost trade-offs across collections of language models \citep{hu2024routerbench,li2026llmrouterbench}.
These approaches primarily treat routing as a selection problem in which a centralized controller assigns a query to one or more passive models. The routed models do not independently choose whether to execute, delegate, or reject a task. Our management formulation shares the objective of assigning simple tasks to inexpensive agents and difficult tasks to more capable agents, but it represents execution and delegation as actions within a Markov game. The resulting policy therefore describes not only which model should be used, but also which agent initiates the delegation and which delegation paths are feasible under the organizational structure 
Recent work has extended routing from single-step model selection to sequential decision making. Router-R1 formulates multi-model routing and aggregation as an RL problem in which an LLM router interleaves internal reasoning with model-invocation actions and receives a cost-aware reward \citep{zhang2025routerr1}. MasRouter expands the routing problem to multi-agent systems by jointly determining a collaboration mode, assigning fixed roles, and selecting an LLM for each role using a cascaded controller \citep{yue2025masrouter}. These methods are particularly close to our cost-performance objective because they adaptively determine which model capabilities should be invoked rather than relying on a fixed cascade.
Nevertheless, the router in these systems remains a centralized orchestrator that constructs or controls the multi-agent workflow. In our framework, delegation is instead endogenous to the participating agents. An agent may exercise its pretrained policy, delegate to another agent, or remain inactive, and the validity and outcome of these choices depend on the actions of the other agents. This formulation supports questions that are not represented by centralized routing, including whether a delegated agent accepts execution, how authority constrains delegation, how reciprocal behavior develops over repeated interactions, and how agents compensate one another for costly actions. 
\subsection{Multi-Agent Reinforcement Learning for LLM Collaboration}
Recent studies have directly formulated LLM collaboration as a cooperative MARL problem. Liu et al.\ introduce Multi-Agent Group Relative Policy Optimization, or MAGRPO, which jointly fine-tunes multiple LLM agents using group-level feedback in multi-agent, multi-turn writing and coding tasks \citep{liu2025magrpo}. Their algorithm addresses the limitation that independently pretrained language models are not necessarily optimized to coordinate and demonstrates that collaborative behavior can be improved through a shared multi-agent training objective.
In another work, CoLLM-CC \citep{liu2026collm} uses a centralized critic during training, whereas CoLLM-DC uses decentralized critics. The comparison shows that Monte Carlo optimization and decentralized value estimation can perform well in relatively short-horizon, dense-reward tasks, while a centralized critic is more effective in long-horizon or sparse-reward settings. This work establishes centralized training with decentralized execution as a promising approach for optimizing LLM-agent collaboration and is directly relevant to the algorithmic implementation of multi-step delegation.
MAGRPO and CoLLM are the closest prior works to ours in their use of MARL for multiple LLM agents. However, the learning targets are different. These methods primarily optimize the agents' language-generation policies within a specified collaboration and execution protocol. Participating agents produce responses or actions according to their assigned roles, and MARL improves how effectively those agents collaborate. Our primary formulation instead learns a meta-policy over agents that already possess pretrained task policies. The meta-policy determines whether a particular pretrained policy should be exercised and whether execution should be delegated to another agent. Thus, MAGRPO and CoLLM optimize behavior within a collaborative arrangement, whereas our framework optimizes the delegation and activation structure governing which underlying capability is used. In other words, our framework can be utilized in conjunction with MAGRPO and CoLLM.
The treatment of heterogeneity also differs. Existing LLM-MARL experiments may pair agents with different model architectures or complementary roles, but they do not primarily formulate heterogeneity as a capability; cost trade-off in which invoking a stronger policy incurs a larger execution cost. Our formulation explicitly associates the underlying agents with different policy quality and operating costs and learns when the additional capability of an expensive agent justifies its use. Moreover, our delegation policies may be restricted by management or delegation relationships, while the LLM-MARL methods above assume a predefined protocol describing which agents participate and how their outputs are combined.
Finally, the two lines of work assign different roles to interaction history. In decentralized LLM collaboration, local trajectories and centralized information are useful for handling partial observability and estimating long-horizon returns. In our history-dependent game, previous actions, delegation choices, and transfers are also strategically meaningful. Agents may condition future behavior on whether another agent previously cooperated, accepted a costly task, or honored a payment. This distinction allows history to support reciprocity and intertemporal incentives even when the current environmental state is fully observable 
Closer to our terminology of \emph{delegation} and \emph{contracts}, Ivanov et al.\ combine deep RL with principal--agent contract theory so a principal can steer agents through outcome-contingent payments \citep{ivanov2024principal}, while Zhu et al.\ propose COMMAND, in which a principal LLM competitively delegates reasoning tasks to multiple agent LLMs \citep{zhu2025command}. Relative to these lines, we focus on cost-aware management and graph-style delegation among frozen heterogeneous solvers, together with history-dependent reciprocity and transferable money inside a Markov-game wrapper.

\section{Algorithm Finding History-dependent Policies $^{\text{\ref{alg:time2-nash-vi}}}$}
\label{app:algorithm}

\begin{algorithm*}[t]
\caption{Two-step Nash Equilibrium Value Iteration}
\label{alg:time2-nash-vi}
\begin{algorithmic}[t]
\footnotesize
\setlength{\abovedisplayskip}{4pt}
\setlength{\belowdisplayskip}{4pt}
\setlength{\abovedisplayshortskip}{2pt}
\setlength{\belowdisplayshortskip}{2pt}

\REQUIRE Initial state $x_1$; state space $\mathcal{S}$; action spaces
$\mathcal{A}_A$ and $\mathcal{A}_B$; transition probabilities
$P(x_2\mid x_1,\bm{a}_1)$; rewards $r_t^i$ and for
$i\in\{A,B\},t\in\{1,2\}$.

\ENSURE A deterministic history-dependent Nash equilibrium policy profile
that maximizes agent $B$'s value.

\STATE Define the set of feasible two-step histories:
\[
\mathcal{H}_2
=
\left\{
(\bm{a}_1,x_2):\bm{a}_1\in\mathcal{A},x_2\in\mathcal{S}
\right\}
\]

\FOR{each history $h_2=(\bm{a}_1,x_2)\in\mathcal{H}_2$}

    \STATE Define the two-step normal-form game by
    \[
    Q^i_2(h_2,\bm{a}_2)
    =
    r^i_2(x_2,\bm{a}_2),
    \qquad i\in\{A,B\}.
    \]

    \STATE Compute the set of pure Nash equilibria
    \[
    \mathcal{E}_2(h_2)
    =
    \left\{
    \bm{a}_2\in\mathcal{A}:
    \begin{array}{l}
    Q^i_2\bigl(h_2,(a^i_2,\bm{a}^{-i}_2)\bigr)
    \geq
    Q^i_2\bigl(h_2,(\widetilde a^i_2,\bm{a}_2^{-i})\bigr),
    \quad \forall \widetilde a^i_2\in\mathcal{A}_i,\forall i\in\{1,2,\dots,N\}\\[1mm]
    \end{array}
    \right\}
    \]

\ENDFOR

\STATE Enumerate all deterministic continuation selectors
\[
\sigma_2:\mathcal{H}_2
\longrightarrow
\mathcal{A}
\]
satisfying
\[
\sigma_2(h_2)\in \mathcal{E}_2(h_2),
\qquad
\forall h_2\in\mathcal{H}_2
\]


\STATE Initialize the set of candidate equilibrium policy profiles:
$\mathcal{C}\leftarrow\varnothing$

\FOR{each continuation selector $\sigma_2\in\mathcal{S}_2$}

    \FOR{each joint action $\bm{a}_1\in\mathcal{A}$ at $t=1$,} 
        \STATE Compute the continuation value
        \[
        W_i^{\sigma_2}(x_1,\bm{a}_1)
        =
        \sum_{x_2\in\mathcal{S}}
        P(x_2\mid x_1,\bm{a}_1)
        r^i_2\!\left(
        x_2,
        \sigma_2(\bm{a}_1,x_2)
        \right),\forall i
        \]

        \STATE Compute the induced utilities at $t=1$
        \[
        Q^{i,\sigma_2}_1(x_1,\bm{a}_1)
        =
        r^i_1(x_1,\bm{a}_1)
        +
        W_i^{\sigma_2}(x_1,\bm{a}_1)
        \]

    \ENDFOR

    \STATE Find the set of pure Nash equilibria of the induced game at $t=1$:
    \[
    \mathcal{E}_1(\sigma_2)
    =
    \left\{
    \bm{a}_1:
    \begin{array}{l}
    Q_1^{i,\sigma_2}\bigr(x_1,(a_1^i,\bm{a}_1^{-i})\bigl)
    \geq
    Q_1^{i,\sigma_2}\bigr(x_1, (\widetilde a_1^i,\bm{a}_1^{-i})\bigl),
    \quad \forall \widetilde a_1^i\in\mathcal{A}_i,\forall i 
    \end{array}
    \right\}
    \]

    \FOR{each $\bm{a}_1\in E_1(\sigma_2)$}

        \STATE Construct the equilibrium policy profile
        \[
        \pi
        =
        \bigl(\bm{a}_1,\sigma_2\bigr)
        \]

        \STATE Set
        \[
        V_i^{\pi}(x_1)
        =
        Q^{i,\sigma_2}_1(x_1,\bm{a}_1),
        \qquad \forall i,\forall x_1\in\mathcal{S},\forall \pi
        \]

        \STATE Add $\pi$ to $\mathcal{C}$

    \ENDFOR

\ENDFOR

\STATE Select
\[
\pi^*
=
\arg\max_{\pi\in\mathcal{C}}\|V_i^\pi\|_1
\]

\RETURN $\pi^*$ and its equilibrium value
$\{V_i^{\pi,*}\}_{i=1}^N)$

\end{algorithmic}
\end{algorithm*}

\section{Proof of Theorem \ref{thm:optimal-selector}}
\label{app:proof}
\begin{proof}
We prove by showing that none of the agents $i$ prefers to deviate from the output policy $\pi^{*}$ at any time step $t=1$ and $t=2$. 

Without the loss of generality, suppose that agent $i$ changes its policy in both steps. Instead of adopting 
\begin{align*}
    &\pi^{i,*}_1(x_1)=a^{i,*}_1 \text{ and }\\
    &\pi^{i,*}_2(x_2,\bm{a}_1)=a^{i,*}_2
\end{align*}
The agent takes 
\begin{align*}
    \widetilde a^{i}_1 \text{ and }
    \widetilde a^{i}_2,\widetilde a_t^i\neq a_t^{i,*}, \text{ and }\widetilde a_t^i\in\mathcal{A}_i,t\in\{1,2\}
\end{align*}
separately in two steps. We denote the deviated joint action at time step $t$ as
$\tilde{\bm{a}}_t=(\widetilde a_t^{i},\bm{a}^{-i}_t)$.

Then, there will be a worse or equal reward for the second time step for agent $i$ for every possible history $h_2$, i.e.,
$$ Q_2^i \bigl(h_2,\tilde{\bm{a}}_2\bigr)\leq Q_2^i\bigl(h_2,\bm{a}^{*}_2\bigr).$$

Then, the continuation values for agent $i$ become worse for every possible first-step action $\bm{a}_1$,
\begin{align}\label{continuation-value}
 &r_2^i(x_2,\tilde{\bm{a}}_2) \leq r_2^i(x_2,\bm{a}_2^*)\\ \nonumber
    \Rightarrow &W_i^{\widetilde \sigma}(x_1,\bm{a}_1) \leq W_i^{\sigma^*}(x_1,\bm{a}_1),
\end{align}
where $\widetilde \sigma(\bm{a}_1,x_2)=\tilde {\bm{a}}_2)$, and $\sigma^*(\bm{a}_1,x_2)=\bm{a}^*_2$.

The agent $i$ also deviates from $a_{1}^{i,*}$ to $\widetilde a_1^i$, then we compare the utilities of it at $t=1$ if it deviates from $\bm{a}_1^*$ found by Algorithm \ref{alg:time2-nash-vi}.
\begin{align*}
    &\qquad Q_1^{i,\sigma^*}(x_1,\bm{a}_1^*)-Q_1^{i,\widetilde\sigma}(x_1,\tilde{\bm{a}}_1)\\&= r_1^i(x_1,\bm{a}_1^*)+W_i^{\sigma^*}(x,\bm{a}_1^*)-r_1^i(x_1,\tilde {\bm{a}}_1) - W_i^{\widetilde\sigma}(x_1,\tilde{\bm{a}_1})\\    &=r_1^i(x_1,\bm{a}_1^*)+W_i^{\sigma^*}(x,\bm{a}_1^*)-W_i^{\widetilde\sigma}(x_1,\bm{a}_1^*) + W_i^{\widetilde\sigma}(x_1,\bm{a}_1^*) \\   
    &-r_1^i(x_1,\tilde{\bm{a}}_1) - W_i^{\widetilde\sigma}(x_1,\tilde{\bm{a}_1})\\
    &=[\bigr(r_1^i(x_1,\bm{a}^*_1)+W_i^{\widetilde\sigma}(x_1,\bm{a}^*_1)\bigl) -\bigl( r_1^i(x_1,\widetilde{\bm{a}}_1) + W_i^{\widetilde\sigma}(x_1,\tilde{\bm{a}_1})\bigr)]  \\&+\bigl( W_i^{\sigma^*}(x,\bm{a}_1^*)-W_i^{\widetilde\sigma}(x_1,\bm{a}_1^*) \bigr)
\end{align*}
Since $\bm{a}_1^*$ is the Nash equilibrium action of the first step found by the algorithm, we know that by adopting the same continuation selector $\widetilde\sigma$, the first big term is greater than $0$, and we already proved that the second term is greater than $0$ in (\ref{continuation-value}). Therefore
$$ Q_1^{i,\widetilde\sigma}(x_1,\tilde{\bm{a}}_1)\leq Q_1^{i,\sigma^*}(x_1,\bm{a}_1^*),$$
and we prove that no agent $i$ wants to deviate from the Nash equilibrium policy $\pi^*$ found by Algorithm \ref{alg:time2-nash-vi}.
\end{proof}

\section{Implementation details for the GSM8K studies}
\label{app:gsm8k-suppl}
\label{app:repro}
This appendix collects protocol details for
Subsections~\ref{subsec:exp21}--\ref{subsec:exp22} that are omitted from the
main text for space.
Nothing here changes the reported numbers; it only records how those numbers
were produced (needed for reproducibility).

\paragraph{Why GSM8K.}
We evaluate on GSM8K~\citep{cobbe2021gsm8k} (\texttt{main} split), a public
grade-school math word-problem benchmark with reference solutions.
We use GSM8K because (i)~free-form answers can be checked against a gold
solution, (ii)~the corpus is large enough to carve out a fixed ability-map
bank, a training pool, and the official test holdout, and (iii)~accuracy
varies sharply across our three frozen solvers, so cost--accuracy trade-offs
are meaningful.

\paragraph{Solvers and synthetic prices.}
We use three frozen solvers: Nemotron~3 Ultra ($A$), Llama~3.2~3B Instruct
($B$), and Llama~3.2~1B Instruct ($C$).
Only the coordinator $\pi$ is trained (centralized DQN).
Synthetic token prices (\$/1M input/output) are $\$1.00/\$3.00$ for Ultra,
$\$0.25/\$1.50$ for Llama~3B, and $\$0.05/\$0.10$ for Llama~1B; these are
experimental parameters, not live provider invoices.

\paragraph{Offline preprocessing.}
Before any RL training, each solver answers eligible GSM8K items once, and we
record measured token usage.
Separately, we ask Nemotron~3 Ultra to mark each stored answer as correct or
incorrect relative to the GSM8K reference solution.
This is an offline labeling step used to measure solver accuracy and to score
cached training episodes.
Listings~\ref{lst:preprocess}--\ref{lst:train} summarize the corresponding
preprocessing and training entrypoints; the full scripts are provided in the
code-and-data supplementary ZIP (not via a web URL).
Using Ultra both as one of the three solvers and as this offline labeler is a
practical choice rather than a modeling requirement.
A separate model (for example GPT-5.5 or GLM-5.2) could have done the same
labeling job; we used Ultra because it was available to us at negligible
additional cost, and because checking free-form answers against the GSM8K
reference is a narrow comparison task.
On a hand-checked sample of $200$ items, Ultra's labels agreed with our
judgments in every case, so we used it for the full offline labeling pass.
On a held-out development bank, the resulting accuracies are about $96.7\%$
for Ultra, $84.5\%$ for Llama~3B, and $53.7\%$ for Llama~1B.
Tokens spent on this offline labeling are excluded from the step cost used by
$\pi$.
From GSM8K \texttt{main}, $1000$ official-train items form a fixed ability map
(empirical success rates by difficulty and solver).
The remaining $6473$ training items are used to train $\pi$, and the official
test set ($1319$) is held out.
Cached answers and correctness labels are reused throughout training so that
policies facing the same question see the same solver outputs.
In particular, we do not call the solvers or a live labeler at every RL step:
$\pi$ reads the already-cached offline answers and labels, which makes the
GSM8K coordinator experiments reproducible.

\paragraph{Training protocol, seeds, and baseline.}
An episode contains $32$ questions.
A question may be retried up to three times; after three failures the
environment advances with penalty $-1$.
We use $\gamma=0.99$, cost weight $\lambda=100$, and $1000$ training episodes.
Each GSM8K study reports mean$\pm$std over $20$ independent runs with seeds
$\{0,1,\ldots,19\}$; for seed $s$ we set the environment, NumPy, and PyTorch
RNGs from $s$ before training.
Figures in the main text use the ability-map observation; a paired no-map run
gives nearly the same aggregates and is summarized in the tables.
Both GSM8K studies use a handoff fee rate $\alpha=0.05$ (defined per subsection).
The \emph{always-Ultra} baseline sends every question to Ultra with no handoff
to $B$ or $C$, under the same caches and the same cost accounting as the
learned policy in that subsection.

\paragraph{Hyper-parameters and selection.}
Table~\ref{tab:gsm8k-dqn-hparams} lists the final coordinator settings shared by
both GSM8K studies (delegation uses $189$ discrete joint actions instead of
$3$).
During development we fixed the environment protocol early (episode length
$32$, three retries, $\lambda=100$, $\alpha=0.05$, $\gamma=0.99$) after small
sanity runs and retained a standard DQN schedule
(Adam learning rate $10^{-3}$, $\varepsilon$ decay $1\!\rightarrow\!0.05$ over
$20{,}000$ gradient steps, replay buffer $80{,}000$, target sync every $200$
updates, hidden size $128$).
We did not run a large grid search; all reported seeds use this single final
configuration.

\begin{table}[t]
\centering
\scriptsize
\setlength{\tabcolsep}{3pt}
\begin{tabular}{lc}
\toprule
Hyper-parameter & Value \\
\midrule
Optimizer / learning rate & Adam / $10^{-3}$ \\
Discount $\gamma$ & $0.99$ \\
Cost weight $\lambda$ / handoff fee $\alpha$ & $100$ / $0.05$ \\
Episodes $\times$ questions/episode & $1000\times 32$ \\
Max attempts per question / fail penalty & $3$ / $-1$ \\
Hidden size / batch / replay buffer & $128$ / $32$ / $80{,}000$ \\
Target sync / $\varepsilon$ schedule & every $200$; $1\!\rightarrow\!0.05$ over $20$k steps \\
Seeds & $0$--$19$ \\
\bottomrule
\end{tabular}
\caption{Final DQN / environment hyper-parameters for the GSM8K management and delegation studies.}
\label{tab:gsm8k-dqn-hparams}
\end{table}

\paragraph{Compute.}
Coordinator training ran on NVIDIA A100-SXM4-80GB GPUs under Linux, using
Python~$3.10$, PyTorch~$2.5.1$ (CUDA~$12.1$ wheels), NumPy, Gymnasium, and
HuggingFace \texttt{transformers} for local Llama solves and caches.
Each seed used a single A100; offline Ultra labeling used API calls and is
excluded from the step cost charged to $\pi$.

\paragraph{Management reward terms.}
In Subsection~\ref{subsec:exp21}, the shared step reward is
$$
r \;=\; R_{\mathrm{corr}}
\;-\; \lambda\, C_{\mathrm{solve}}(a_{\mathrm{exec}})
\;-\; \lambda\, C_{\mathrm{edge}}(n_e,a_{\mathrm{exec}})
\;+\; R_{\mathrm{fail}}.
$$
Here $R_{\mathrm{corr}}\in\{0,1\}$ is one if the chosen solver's cached answer
is correct under the offline Ultra labels above, and zero otherwise;
$C_{\mathrm{solve}}(a_{\mathrm{exec}})$ is that solver's synthetic token cost;
$n_e$ is the number of handoffs on the chosen route ($0$, $1$, or $2$);
$C_{\mathrm{edge}}(n_e,a_{\mathrm{exec}})=n_e\,\alpha\,C_{\mathrm{solve}}(a_{\mathrm{exec}})$
with $\alpha=0.05$ is a small fee for each handoff; $\lambda=100$ scales the
dollar costs into the reward; and $R_{\mathrm{fail}}=-1$ only after three failed
attempts on the same question.
In short, we pay for answering and also pay a little for each pass down the
chain.

\paragraph{Delegation cost accounting.}
In Subsection~\ref{subsec:exp22}, we count the attempt as successful if any
answering agent is correct under the offline Ultra labels, and we charge the
sum of those agents' solve costs.
In addition, each kickoff handoff and each peer pass pays a small edge fee of
$\alpha=0.05$ times the recipient's expected solve cost, matching the handoff
fee used in the management study.
Always-Ultra is evaluated in the same graph environment, so it also pays the
kickoff edge fee when Ultra is started.

\paragraph{Code and entrypoints.}
Listings~\ref{lst:preprocess}--\ref{lst:train} summarize the offline pipeline
and DQN entrypoints; the full scripts are in the code-and-data supplementary ZIP.

\begin{listing}[t]
\begin{lstlisting}[language=Python,basicstyle=\ttfamily\scriptsize]
# Offline GSM8K preprocessing (conceptual; full scripts in code release)
items = load_gsm8k_main()                 # Cobbe et al., GSM8K main
label, train, test = split(items,
    label_n=1000, seed=0)                 # test = official 1319
for solver in [Ultra, Llama3B, Llama1B]:
    cache[solver] = answer_once(label + train + test)  # store tokens
for (q, ans) in all_cached_pairs(cache):
    judge[q, ans] = Ultra_marks_vs_gold(q, ans)        # offline only
ability_map = success_rates(label, judge)              # difficulty x solver
# RL samples from train caches; judge tokens excluded from step cost
\end{lstlisting}
\caption{Offline preprocessing used before both GSM8K RL studies.}
\label{lst:preprocess}
\end{listing}

\begin{listing}[t]
\begin{lstlisting}[language=bash,basicstyle=\ttfamily\scriptsize]
# Coordinator training entrypoints (seed S in 0..19; PYTHONPATH=src)
python experiments/exp21_management/run_train.py \
  --config configs/exp21/cache_dqn_v1_run3.yaml --seed S
python experiments/exp22_delegation/run_train.py \
  --config configs/exp22/dqn_v2_edge_fee_with_map.yaml --seed S
# Warm/label helpers: run_warm_cache.py, run_label.py,
# run_build_ability_map_local.py (same experiments/ tree)
\end{lstlisting}
\caption{Training and preprocessing entrypoints for the GSM8K studies.}
\label{lst:train}
\end{listing}

\end{document}